\documentclass[12pt]{article}
\usepackage{bbm}
\usepackage{latexsym}
\usepackage{verbatim}
\usepackage{mathrsfs}
\usepackage{amsmath}
\usepackage{graphicx}
\usepackage{booktabs} %比较表格
\usepackage{amssymb}
\usepackage{color}
\usepackage{amsthm}
\usepackage{cite}
\usepackage{indentfirst}
\usepackage{anysize}\marginsize{45mm}{45mm}{40mm}{50mm}
\usepackage{multirow}  %表格多行支持
\usepackage[ruled,vlined]{algorithm2e} %算法环境支持
\usepackage{caption}
\usepackage{enumerate}

\newtheoremstyle{lemma}{\topsep}{\topsep}%
     {}%         Body font
     {}%         Indent amount (empty = no indent, \parindent = para indent)
     {\bfseries}% Thm head font
     {.}%        Punctuation after thm head 定理标号后的点
     {0.4em}%     Space after thm head (\newline = linebreak)定理标号点之后的距离
     {\thmname{#1}\thmnumber{ #2}\thmnote{ #3}}%         Thm head spec
\theoremstyle{lemma}  

\theoremstyle{remark}

\newtheorem{theorem}{Theorem}              %ȫ ֶ   
\newtheorem{lemma}[theorem]{Lemma}
\newtheorem{corollary}[theorem]{Corollary}

\numberwithin{equation}{section}

\usepackage{latexsym, bm}
\begin{document}
\title{Fault-tolerant Hamiltonian connectivity of Johnson graphs \thanks{This research was partially supported by the National Natural Science Foundation of China (No. 11801061)}}

%% author 直到邮箱结束
\author{Huazhong L\"{u}\thanks{Corresponding author.} and Jinhao Liu \\
{\small School of Mathematical Sciences,} \\
{\small University of Electronic Science and Technology of China,} \\
{\small Chengdu, Sichuan 611731, P.R. China}\\
{\small E-mails: lvhz@uestc.edu.cn; 202321110137@std.uestc.edu.cn}\\}
\date{}

\maketitle

\begin{abstract}
Johnson graphs $J(n,k)$ are a classical family of highly symmetric networks known to be Hamiltonian-connected in the fault-free setting. In this paper, we investigate their Hamiltonian connectivity under three failure models, namely general edge faults, matching faults, and vertex faults. For general edge faults, we prove that $J(n,k)$ remains Hamiltonian-connected after the deletion of any set of at most $k(n-k)-3$ edges for $n\geq4$. Since $J(n,k)$ is $k(n-k)$-regular, this attains the natural degree-based upper bound for Hamiltonian connectivity. We then consider matching faults, which exclude the concentration of multiple faulty links at a single vertex and permit substantially larger fault sets. We show that $J(n,k)$ remains Hamiltonian-connected after the deletion of an arbitrary matching for $n\geq5$, including a perfect matching whenever one exists. For vertex failures, we prove that $J(n,k)$ is $(n-2)$-vertex-fault-tolerant Hamiltonian-connected for $n\geq5$. All three results are constructive and lead to recursive fault-tolerant Hamiltonian routing algorithms. Simulation results on Johnson graphs with up to $12{,}870$ vertices further show that the routing algorithms successfully construct
fault-free Hamiltonian paths for all tested source-destination pairs, with measured execution times exhibiting near-linear growth with network size. These results establish a unified fault-tolerant Hamiltonian-connectivity framework for Johnson graphs under different failure patterns.

\vskip 0.1 in


\end{abstract}

\section{Introduction}
Johnson graphs form an important family of highly symmetric graphs with rich combinatorial and algebraic structures. For integers $n$ and $k$ with $1\leq k<n$, the Johnson graph $J(n,k)$ has all $k$-subsets of $[n]=\{1,2,\ldots,n\}$ as its vertices, where two vertices are adjacent whenever their intersection has size $k-1$. Thus, $J(n,k)$ has $\binom{n}{k}$ vertices and is $k(n-k)$-regular. Owing to their vertex-transitivity and distance-regularity \cite{CG}, Johnson graphs have received considerable attention in algebraic
graph theory \cite{Brouwer,Koshelev,Koshelev2}, combinatorial and Hamiltonian properties \cite{Alspach2,Kozhevnikov}, and studies of highly symmetric network topologies \cite{LL,Park0}.

A wide range of algorithms for high-performance computing rely on processor organizations with linear arrays and rings. These structures therefore provide natural models for communication, control, and data organization in parallel and distributed systems \cite{Besta}. In graph-theoretic terms, a linear array spanning all processors corresponds to a {\em Hamiltonian path}, which visits every vertex exactly once. A network is said to be {\em Hamiltonian-connected} if such a spanning path exists between every pair of distinct vertices. The use of Hamiltonian paths in multicast routing can help reduce the likelihood of deadlocks and alleviate network congestion \cite{Ebrahimi}. In addition, Ye and Liang \cite{Ye} proposed a five-round adaptive fault-diagnosis algorithm for distributed-memory multiprocessors that relies on a Hamiltonian interconnection topology. Consequently, Hamiltonian properties have been extensively studied for a variety of interconnection networks. Among them, Hamiltonian connectivity represents a particularly strong spanning-path property and has been established for hypercube variants \cite{Dong,Guo,Hsieh2007,Huang,C.Park}, $(n,k)$-star graphs \cite{HsuHC2003}, arrangement graphs \cite{HsuHC2004}, and cycle composition networks \cite{KLLTH08}. For Johnson graphs, Alspach \cite{Alspach} proved that $J(n,k)$ is Hamiltonian-connected, establishing this strong spanning-path property in the fault-free setting.

In practical networks, however, processors and communication links may fail, making it important to determine whether Hamiltonian connectivity can be maintained in the presence of faults. Under general edge faults, the locations of faulty links are unrestricted, and fault-tolerant Hamiltonian properties under this model have been extensively investigated for various interconnection networks \cite{HaoRX2014,Lan,LiP2017,TsaiCH2002,ZhouQ2015}. Consequently, determining the maximum number of general edge faults that can be tolerated while preserving Hamiltonian connectivity is a fundamental problem in the reliability analysis of interconnection networks.

Beyond the unrestricted fault model, a variety of constrained edge-fault models that exclude certain local obstructions or restrict the distribution of faulty links have also been investigated, including conditional fault models \cite{Ho2009,Hsieh2007,Hsieh2016}, in which a prescribed minimum residual degree is required, as well as more recent partitioned-edge \cite{Zhuang} and region-based fault models \cite{Wang}. These studies show that incorporating information about the distribution of faults can yield stronger fault-tolerance guarantees for Hamiltonian connectivity.

A natural way to accommodate a substantially larger number of link failures is to restrict the faulty edges to be pairwise disjoint, i.e., a matching. Such a restriction prevents multiple faulty links from concentrating at a single vertex. Hamiltonian properties under matching faults have therefore received considerable attention in hypercube-based networks. Dimitrov et al. \cite{Dimitrov} characterized Hamiltonian paths and cycles in hypercubes avoiding a prescribed matching in terms of certain forbidden half-layer and almost-layer structures. Dybizba\'{n}ski and Szepietowski \cite{Dybizbanski} further studied hypercubes with pairwise disjoint faulty edges and showed that, for $n\geq4$, $Q_n$ is Hamiltonian if and only if every dimension contains two fault-free crossing edges of different parity. For the balanced hypercube $BH_n$, Lan and L\"u \cite{Lan2} proved that every fault-free edge lies on a fault-free Hamiltonian cycle even when the faulty-edge set is a perfect matching. Similarly, Zuo and L\"u \cite{Zuo} showed that, for any matching $F$ of the folded hypercube $FQ_n$, $FQ_n-F$ remains Hamiltonian for $n\geq2$ and is Hamiltonian laceable for every odd $n\geq3$. In particular, in both $BH_n$ and $FQ_n$, the number of admissible faulty edges can reach $2^{2n-1}$ and $2^{n-1}$, respectively, both exponential in the network dimension $n$. 

Whereas an edge fault disables a single communication link, a vertex fault removes a processor together with all of its incident links and can therefore affect the network more extensively. Accordingly, vertex failures have been studied both independently and together with edge failures in mixed fault models. In particular, fault-tolerant Hamiltonian properties under combined vertex and edge failures have been established for several interconnection networks, including $(n,k)$-bubble sort graphs \cite{LiuTian2026}, 2-tree-generated networks \cite{AbdallahCheng2022}, and arrangement graphs \cite{HsuHC2004}. These results further motivate the study of Hamiltonian connectivity under vertex failures in Johnson graphs.

For bipartite networks, Hamiltonian laceability serves as the natural counterpart of Hamiltonian connectivity because the bipartition imposes parity restrictions on the endpoints of a Hamiltonian path. To provide a comparison with representative results for several well-known interconnection networks, Table \ref{tab:comparison} summarizes fault-tolerant Hamiltonian properties under different failure models.

\begin{table}[htbp]
    \centering
    \caption{Comparison of fault-tolerant Hamiltonian properties of representative interconnection networks.}
    \label{tab:comparison}
    \renewcommand{\arraystretch}{1.2}
    \resizebox{\columnwidth}{!}{%
    \begin{tabular}{l c l l l}
        \toprule
        \textbf{Network} & \textbf{Bipartite} & \textbf{Hamiltonian Property} & \textbf{Fault Model} & \textbf{Fault-Tolerance Result}  \\
        \midrule
        \multirow{2}{*}{Hypercube $Q_n$} & \multirow{2}{*}{Yes} & Hamiltonian Laceability & General Edge & $n-2$ \cite{TsaiCH2002} \\
         & & Hamiltonicity & Matching & Dimension-parity condition \cite{Dybizbanski} \\
        \midrule
        \multirow{2}{*}{Balanced Hypercube $BH_n$} & \multirow{2}{*}{Yes} &   Hamiltonian Laceability & General Edge & $2n-2$ \cite{ZhouQ2015} \\
         & & Hamiltonicity & Matching & general matching \cite{Lan2} \\
        \midrule
        \multirow{2}{*}{Folded Hypercube $FQ_n$ (odd $n$)} & \multirow{2}{*}{Yes} & Strongly Hamiltonian Laceability & General Edge & $n-1$ \cite{Hsieh2008} \\
         & & Hamiltonian Laceability & Matching & Arbitrary matching \cite{Zuo} \\
        \midrule
        Star Graph $S_n$ & Yes & Hamiltonian Laceability & General Edge & $n-3$ \cite{LiT2002} \\
        \midrule
        $k$-ary $n$-cube ($k\geq4$ even) & Yes & Hamiltonian Laceability & General Edge & $2n-2$ \cite{Stewart} \\
        \midrule
        $(n,k)$-Star Graph $S_{n,k}$ ($n-k\geq2$) & No & Hamiltonian Connectivity & Mixed Vertex/Edge & $|F_v|+|F_e|\leq n-4$ \cite{HsuHC2003} \\
        \midrule
        Arrangement Graph $A_{n,k}$ & No & Hamiltonian Connectivity & Mixed Vertex/Edge & $|F_v|+|F_e|\leq k(n-k)-3$ \cite{HsuHC2004} \\
        \midrule
        $(n,k)$-Bubble Sort Graph & No & Hamiltonian Connectivity & Mixed Vertex/Edge & $|F_v|+|F_e|\leq n-4$ \cite{LiuTian2026}  \\
        \midrule
        2-tree-generated network & No & Hamiltonian Connectivity & Mixed Vertex/Edge & $|F_v|+|F_e|\leq 2n-7$ \cite{AbdallahCheng2022} \\
        \midrule
        \multirow{3}{*}{\textbf{Johnson} $J(n,k)$} & \multirow{3}{*}{No} & \multirow{3}{*}{\textbf{Hamiltonian Connectivity}} & \textbf{General Edge} & $k(n-k)-3$ \\
         & & & \textbf{Matching} & Arbitrary matching \\
         & & & \textbf{Vertex} & $n-2$  \\
        \bottomrule
    \end{tabular}%
}
\parbox{\columnwidth}{\footnotesize
\textit{Note:} $F_v$ and $F_e$ denote the sets of faulty vertices and faulty edges, respectively.}
\end{table}

Despite the extensive literature on fault-tolerant Hamiltonian properties of well-known interconnection networks, corresponding results for Johnson graphs remain relatively limited. In particular, the following questions arise naturally. First, how many arbitrarily distributed faulty edges can $J(n,k)$ tolerate while remaining Hamiltonian-connected? Second, if the faulty edges are required to form a matching, can Hamiltonian connectivity be preserved for fault sets far exceeding this general-fault bound, possibly even for a perfect matching? Finally, how robust is the Hamiltonian connectivity of Johnson graphs against vertex failures? In this paper, we answer these questions constructively. Our main contributions are summarized as follows.

\begin{itemize}
\item \textbf{General edge faults.}
We prove that $J(n,k)-F$ remains Hamiltonian-connected for every edge-fault set $F$ satisfying $|F|\leq k(n-k)-3$.
Since $J(n,k)$ is $k(n-k)$-regular, this bound reaches the natural upper bound imposed by the minimum-degree obstruction.

\item \textbf{Matching faults.}
We establish that $J(n,k)$ remains Hamiltonian-connected after the deletion of an arbitrary matching. In particular, the result still holds when the faulty edges form a perfect matching whenever one exists. Thus, under this model, the number of faulty links can reach $\frac{1}{2}\binom{n}{k}$, i.e., half the order of $J(n,k)$.

\item \textbf{Vertex faults.}
We show that $J(n,k)$ is $(n-2)$-vertex-fault-tolerant Hamiltonian-connected; that is, after the removal of any set of at most $n-2$ faulty vertices, each pair of remaining fault-free vertices can still be joined by a Hamiltonian path covering all fault-free vertices.

\item \textbf{Constructive routing algorithms.}
The constructive proofs lead directly to recursive fault-tolerant Hamiltonian routing algorithms for the three fault models. We further analyze their complexity and evaluate their practical execution time through simulations on Johnson graphs of different scales.
\end{itemize}

These results establish a systematic fault-tolerant Hamiltonian-connectivity framework for Johnson graphs and demonstrate their strong resilience in supporting Hamiltonian routing under diverse failure scenarios.

The remainder of this paper is organized as follows. Section II introduces the necessary notations and basic properties of Johnson graphs, and then investigates Hamiltonian connectivity under general edge faults. Section III establishes the fault-tolerant Hamiltonian connectivity under arbitrary matching faults. Section IV studies vertex faults and derives the corresponding vertex-fault-tolerant Hamiltonian-connectivity result. Section V develops the Hamiltonian routing algorithms associated with the three fault models, together with their complexity analysis and performance evaluation. Finally, Section VI concludes the paper.

\section{Optimal Edge-fault-tolerant Hamiltonian Connectivity of $J(n,k)$}

We first introduce some notation and preliminary results that will be used throughout this paper. For a fixed element $i \in [n]$, let $X(i)$ denote the subgraph of $J(n,k)$ induced by the vertices containing $i$, and let $Y(i)$ denote the subgraph induced by the vertices not containing $i$. When the choice of $i$ is clear from the context, we simply write $X$ and $Y$. Clearly,
$X \cong J(n-1, k-1)$ and $Y \cong J(n-1, k)$.
For an edge-fault set $F \subseteq E(J(n,k))$, let $F_X$ and $F_Y$ denote the sets of faulty edges contained in $X$ and $Y$, respectively. We also recall that $J(n,k) \cong J(n,n-k)$ for $0 \le k \le n$, and $J(n,1) \cong K_n$.

\begin{lemma}\cite{Alspach}\label{HC_Johnson} 
For $0\le k\le n$ and $n\ge1$, $J(n,k)$ is Hamiltonian-connected.
\end{lemma}

\begin{lemma} \cite{LL}\label{P2C_Johnson} 
$J(n,k)$, for $n\ge 4$ and $1\le k \le n-1$, is paired $2$-coverable.
\end{lemma}

\begin{lemma}\cite{Ho2009}\label{CHC_K_n}
Let $F\subseteq E(K_n)$ satisfy $\delta(K_n-F)\ge3$. Then $K_n-F$ is Hamiltonian-connected whenever $|F|\le 2n-10$ for $n\notin\{4,5,8,10\}$; the corresponding maximum bounds for $n=4,5,8,10$ are $0,2,5,$ and $9$, respectively.
\end{lemma}

As a direct consequence of Lemma \ref{CHC_K_n}, we obtain the following edge-fault Hamiltonian-connectivity result for complete graphs.
    
\begin{corollary}\label{FT_HC_Kn}
For $n \ge 4$, $K_n$ is $(n-4)$-edge-fault-tolerant Hamiltonian-connected.
\end{corollary}

\begin{proof}
Let $F \subseteq E(K_n)$ be any faulty edge set with $|F| \le n-4$. Since every vertex of $K_n$ has degree $n-1$,
\[
\delta(K_n-F)\ge (n-1)-|F|
\ge (n-1)-(n-4)=3.
\]
Thus, the minimum-degree condition in Lemma \ref{CHC_K_n} is satisfied.

It remains to verify that the number of faulty edges is within the corresponding bound given in Lemma \ref{CHC_K_n}. If $n\ge 6$ and $n\notin\{8,10\}$, then $|F|\le n-4\le 2n-10$.
For the exceptional cases $n\in\{4,5,8,10\}$, we have $|F|\le 0,\,1,\,4,\,6$, respectively, which do not exceed the corresponding bounds $0,\,2,\,5,\,9$
given in Lemma \ref{CHC_K_n}. Therefore, $K_n-F$ is Hamiltonian-connected.
\end{proof}

We now turn our attention to the Johnson graphs. For small parameters, although the graphs are relatively small, manually checking all possible faulty-edge sets and all pairs of endpoints is still combinatorially tedious. Therefore, we developed an exhaustive-search program to rigorously verify the base cases. For each admissible faulty-edge set 
$F$, we enumerate all pairs of distinct vertices $u,v\in V(J(n,k)-F)$ and exhaustively verify the existence of a Hamiltonian $u$-$v$ path in $J(n,k)-F$.

\begin{lemma}\label{FT_HC_J42}
$J(4,2)$ is 1-edge-fault-tolerant Hamiltonian-connected.
\end{lemma}

Since $J(5,2) \cong J(5,3)$, by computationally verifying $J(5,2)$, we immediately obtain the following lemma.

\begin{lemma}\label{FT_HC_J52}
Both $J(5,2)$ and $J(5,3)$ are $3$-edge-fault-tolerant Hamiltonian-connected.
\end{lemma}

Before presenting our main result, we first note a natural upper
bound for edge-fault Hamiltonian connectivity. Let $G$ be
a graph with minimum degree $\delta$. If $\delta-2$ edges incident
with a vertex $x$ fail, then $x$ has only two fault-free neighbors,
say $u$ and $v$. Any Hamiltonian path joining $u$ and $v$ would
necessarily contain the subpath $\langle u,x,v\rangle$, preventing
the path from spanning the remaining vertices. Hence, the ordinary
edge-fault Hamiltonian-connectivity bound of $G$ cannot exceed
$\delta-3$. Since $J(n,k)$ is $k(n-k)$-regular, its edge-fault
Hamiltonian-connectivity bound is at most $k(n-k)-3$. We next show
that this upper bound is attained by Johnson graphs.
    
\begin{theorem}\label{FT_HC_Johnson}
For $n\ge 4,1\le k \le n-1$, $J(n,k)$ is $(k(n-k)-3)$-edge-fault-tolerant Hamiltonian-connected. 
\end{theorem}

\begin{proof} Let $F$ be a set of faulty edges in $J(n, k)$ satisfying $|F| \le (n-k)k - 3$. It follows from Corollary \ref{FT_HC_Kn}, and Lemmas \ref{FT_HC_J42} and \ref{FT_HC_J52},
  that the statement holds for $J(n,1)$, $J(n,n-1)$ with $n \ge 4$, and for $J(4,2)$, $J(5,2)$ and $J(5,3)$. 

When considering $J(n,k)$ with $n\ge 6$ and $2\le k \le n-2$, the induction hypotheses are: $J(m,k')$ is $((m-k')k'-3)$-edge-fault-tolerant Hamiltonian-connected whenever $1 < k' < k$ and $m > k'$, or $J(m,k)$ is $((m-k)k-3)$-edge-fault-tolerant Hamiltonian-connected whenever $m < n$ and $k < m$.

\noindent\textbf{Case 1.} $n<2k$. Then $n-k<k$. Since $J(n,k)\cong J(n,n-k)$, the induction hypothesis applied to $J(n,n-k)$ yields that $J(n,n-k)$ is $((n-k)k-3)$-edge-fault-tolerant Hamiltonian-connected. Hence, $J(n,k)$ has the same property.

\noindent\textbf{Case 2.} $n \ge 2k$. Let
\[
    A=\left\{i\in [n]:\left|F_{X(i)}\right|\ge (n-k)(k-1)-2\right\}.
\]
We claim that $|A|\le k$. Suppose on the contrary that $|A|\ge k+1$. Then
\[
    \sum_{i\in A}\left|F_{X(i)} \right| \ge [(n-k)(k-1)-2](k+1).
\]
Since the two endpoints of each edge share $k-1$ common elements, the sum $\sum_{i\in A}\left|F_{X(i)}\right|$ counts each edge at most $k-1$ times. Thus,
\begin{align*}
    |F|&\ge \frac{1}{k-1}\sum_{i\in A}\left|F_{X(i)} \right|\\
    &\ge [(n-k)(k-1)-2]\frac{k+1}{k-1}\\
    &=(n-k)k + n - \left(k-1+\frac{4}{k-1}+3\right).
\end{align*}
Since $2\le k \le \dfrac{n}{2}$, we have
\[
    k-1+\frac{4}{k-1}+3 \le \max \{8,\frac{n}{2}+\frac{8}{n-2}+2 \}.
\]
Note that $n\ge 6 $, so we have
\[
    |F|\ge (n-k)k -2 > (n-k)k -3.
\]
Hence, the claim holds.

Given that $|A|\leq k$, define $B=[n]\setminus A$. Consequently, $|B| \ge n-k$. Next, we demonstrate by contradiction that $B$ must contain at least one element $j$ such that $|F_{Y(j)}| \le (n-k-1)k-3$, except for the case where $n=6$ and $k=|A|=|B|=3$. Suppose no such $j$ exists in $B$. Then,
\[
    \sum_{j\in B}\left|F_{Y(j)} \right| \ge [(n-k-1)k-2]\cdot|B|.
\]
If an edge is contained in $F_{Y(j)}$, then the two endpoints of the edge do not contain $j$. Since the union of the $k$-sets corresponding to the endpoints of one edge contains $k+1$ elements, the sum $\sum_{j\in B}\left|F_{Y(j)} \right|$ counts each edge at most $n-k-1$ times. Then
\begin{align*}
    |F|&\ge \frac{1}{n-k-1}\sum_{j\in B}\left|F_{Y(j)} \right|\\
    &\ge \frac{|B|}{n-k-1} [(n-k-1)k-2] \\
    &\ge (n-k)k - \frac{2(n-k)}{n-k-1}.
\end{align*}
Note that $2\le k \le \dfrac{n}{2}$ and $n\ge 6$. \\
(1) If $n\ne 6$ or $k\ne 3$, we have $n-k > 3$, which implies $\frac{2(n-k)}{n-k-1} < 3$. Thus, 
\[
    |F|\ge (n-k)k - \frac{2(n-k)}{n-k-1} > (n-k)k - 3,
\]
which contradicts the assumption that $|F| \le (n-k)k - 3$. \\
(2) If $n=6$ and $k=3$, we have $n-k = 3$, which yields
\[
    |F|\ge (n-k)k - \frac{2(n-k)}{n-k-1}=(n-k)k - 3.
\]
This implies that $|F|=6$ and $|B| \ge n-k = 3$. If $|B|>3$ (i.e., $|B| \ge 4$), then
\[
    |F|\ge \frac{4}{5-k}[(5-k)k-2] = 8 > 6,
\]
which is impossible since $|F| \le 6$.

Therefore, except for the case $n=6$ and $k=|A|=|B|=3$, there exists an element $j\in B$ such that
\[
|F_{X(j)}|\le (n-k)(k-1)-3
\quad\text{and}\quad
|F_{Y(j)}|\le (n-k-1)k-3.
\]
Fix such an element $j$, and write $X=X(j)$ and $Y=Y(j)$ in what follows.

Let $u$ and $v$ be any two distinct vertices of $J(n,k)$. We distinguish the following cases according to the locations of $u$ and $v$ in the induced subgraphs $X$ and $Y$. For the case $n=6$ with $k=|A|=|B|=3$, if there exists an element $j$ such that the faulty edge sets of the two induced subgraphs $X$ and $Y$ determined by $j$ satisfy the above inequality, then the desired result follows directly by induction. Otherwise, Lemma \ref{FT_HC_J63_Special} establishes that $J(6,3)$ remains $6$-edge-fault Hamiltonian-connected. 

\noindent\textbf{Case 2.1.} $u, v \in X$. By the induction hypothesis, there exists a Hamiltonian path joining $u$ and $v$ in $X-F$, which we denote by $\langle u=x_1, x_2, \ldots, x_{t-1}, x_t = v \rangle$, where $t=\binom{n-1}{k-1}$. For $n\ge 6$ and $n\ge 2k \ge 4$, we have
\begin{align*}
    &\left\lfloor\dfrac{\binom{n-1}{k-1}}{2}\right\rfloor (n-k)-[(n-k)k-3]\\
    \ge & \frac{1}{2}\binom{n-1}{k-1}(n-k) - \frac{1}{2}(n-k)-(n-k)k+3\\
    \ge & \frac{1}{2}\left[\binom{n-1}{k-1}-2k-1\right](n-k) + 3\\
    \ge & 3.
\end{align*}
This implies that there exist two non-faulty edges $x_ly'$ and $x_{l+1}y''$ such that $y', y'' \in V(Y)$ and $y' \neq y''$. By the induction hypothesis, there exists a Hamiltonian path joining $y'$ and $y''$ in $Y-F$, denoted by $\langle y', P_1, y'' \rangle$. Then $\langle u, \ldots,x_l,y', P_1 , y'',$ $x_{l+1},\ldots, v \rangle$ is a Hamiltonian path joining $u$ and $v$ in $J(n,k)-F$.

\noindent\textbf{Case 2.2.} $u, v \in Y$. By the induction hypothesis, there exists a Hamiltonian path joining $u$ and $v$ in $Y-F$, denoted by $\langle u=y_1, y_2, \ldots, y_{s-1}, y_s = v \rangle$, where $s=\binom{n-1}{k}$. For $n\ge 6$ and $n\ge 2k \ge 4$, we have
\begin{align*}
    &\left\lfloor\dfrac{\binom{n-1}{k}}{2}\right\rfloor k-[(n-k)k-3]\\
    \ge & \frac{1}{2}\binom{n-1}{k}k - \frac{1}{2}k-(n-k)k+3\\
    \ge & \frac{1}{2}\left[\binom{n-1}{k}-2n+2k-1\right]k + 3\\
    \ge & 3.
\end{align*}
This implies that there exist two non-faulty edges $x'y_m$ and $x''y_{m+1}$ such that $x', x'' \in V(X)$ and $x' \neq x''$. By the induction hypothesis, there exists a Hamiltonian path joining $x'$ and $x''$ in $X-F$, denoted by $\langle x', P_2, x'' \rangle$. Then $\langle u, \ldots,y_m,x',P_2,x'',y_{m+1},$ $\ldots,v \rangle$ is a Hamiltonian path joining $u$ and $v$ in $J(n,k)-F$.\\
\noindent\textbf{Case 2.3.} $u \in X$ and $v \in Y$. Let $S = \{ab\in E(J(n,k)):a\in X\setminus\{u\},b\in Y\setminus\{v\}\}$. Then
 \[
     |S|\ge \binom{n-1}{k-1}(n-k)-n > (n-k)k - 3 \geq |F|.
 \]
 This implies that there exists an edge $ab \in S - F$. By the induction hypothesis, there exists a Hamiltonian path joining $u$ and $a$ (resp.\ $v$ and $b$) in $X - F$ (resp.\ $Y - F$),  denoted by $\langle u, P_3, a \rangle$ (resp.\ $\langle b, P_4, v \rangle$). Clearly, the path $\langle u, P_3, a, b, P_4, v \rangle$ is a Hamiltonian path joining $u$ and $v$ in $J(n,k) - F$. 
 \end{proof}

In the following lemma, we deal with the exceptional case from Theorem 7.

\begin{lemma}\label{FT_HC_J63_Special}
Let $F\subseteq E(J(6,3))$ with $|F|\le6$, and define
\[
A=\{i\in[6]:|F_{X(i)}|\ge4\},\qquad B=[6]\setminus A.
\]
Suppose $|A|=|B|=3$ and $|F_{Y(j)}|\ge4$ for every $j\in B$.
Then $J(6,3)-F$ is Hamiltonian-connected.
\end{lemma}

\begin{proof}

Since two adjacent vertices of $J(6,3)$ share exactly two elements,
each faulty edge is counted exactly twice in
$\sum_{i=1}^{6}|F_{X(i)}|$. Hence,
\[
2|F|
=
\sum_{i=1}^{6}|F_{X(i)}|
\geq
\sum_{i\in A}|F_{X(i)}|
\geq 3\cdot 4=12.
\]
Since $|F|\leq 6$, it follows that $|F|=6$. Moreover, equality holds
throughout. Consequently,
\[
|F_{X(i)}|=4
\qquad\text{for every }i\in A,
\]
and no faulty edge is counted by $F_{X(j)}$ for any $j\in B$.
Thus, the two common elements of the endpoints of every faulty edge
belong to $A$.

Without loss of generality, let $A = \{1, 2, 3\}$ and $B = \{4, 5, 6\}$. For convenience, let $CE_t$ ($t = 1, 2, \cdots, 6$) denote the set of elements common to the two endpoints of the $t$-th faulty edge in $F$. By the definition of $A$, among the six sets $CE_t$, at least four contain the element $1$, at least four contain the element $2$, and at least four contain the element $3$. Since the total number of elements in the six sets $CE_t$ is exactly 12, it follows that there are exactly two sets containing $\{1,2\}$, exactly two sets containing $\{1,3\}$, and exactly two sets containing $\{2,3\}$. On the other hand, by the assumption, for any $j \in B$, we have $|F_Y(j)| \ge 4$. Since $j$ cannot be a common element of the two endpoints of a faulty edge, each $j \in B$ can occur in at most two endpoints of faulty edges. Therefore, the number of endpoints containing elements from $\{4,5,6\}$ is at most six. Consequently, at least six endpoints correspond to vertices containing the elements $\{1,2,3\}$. Since all these vertices are endpoints of faulty edges, the 12 endpoints must consist of six copies of $\{1,2,3\}$, two vertices containing $4$, two vertices containing $5$, and two vertices containing $6$. It is straightforward to verify that all such configurations are isomorphic to one another. Hence, we may assume that $a=\{1,2,3\}$, $b_1 = \{1, 2, 4\}$, $b_2 = \{1, 2, 5\}$, $b_3 = \{1, 3, 5\}$, $b_4 = \{1, 3, 6\}$, $b_5 = \{2, 3, 6\}$, $b_6 = \{2, 3, 4\}$, and $F=\{ab_1,ab_2,ab_3,ab_4,ab_5,ab_6\}$.

Let $u$ and $v$ be any two distinct vertices of $J(6, 3)$. We now discuss the following cases based on the distribution of these vertices in $X(4)$ and $Y(4)$. It is clear that $X - F = X$.

\noindent\textbf{Case 1.} $u, v \in X$.
Then $N_{J(6,3)-F}(a) \cap V(X)$ has only one vertex $b_7 = \{1, 3, 4\}$.

\noindent\textbf{Case 1.1.} $b_7 \in \{u, v\}$, say $u = b_7$.
Note that $b_2 = \{1,2,5\} \in V(Y)$. Let $F_1 = F \setminus \{ab_2\}$. Since $F$ contains exactly four edges entirely within $Y$, $F_1$ leaves exactly three faulty edges in $Y$. By Lemma \ref{FT_HC_J52}, $J(5,3)$ is 3-edge-fault-tolerant, so there exists a Hamiltonian path in $Y - F_1$ joining $a$ and $b_2$, denoted by $\langle a, Q_1, b_2 \rangle$. Choose a vertex $b_2' \in N_X(b_2) \cap X$ such that $b_2' \ne v$. Since $X \cong J(5,2)$ is a $6$-regular graph on ten vertices, the vertex-deleted subgraph $X-\{u\}$ has nine vertices with a minimum degree of five. By the classic Ore-type condition for Hamiltonian connectivity ($d(x)+d(y) \ge 5+5 = 10 \ge 9+1$), $X- \{u\}$ remains Hamiltonian-connected. Thus, there exists a Hamiltonian path in $X - \{u\}$ joining $b_2'$ and $v$, denoted by $\langle b_2', P_1, v \rangle$. Then the path $\langle u, a, Q_1, b_2, b_2', P_1, v \rangle$ is a Hamiltonian path in $J(6, 3) - F$ joining $u$ and $v$.

\noindent\textbf{Case 1.2.} $b_7 \notin \{u, v\}$.
Note that $b_2 = \{1,2,5\}$ has exactly three neighbors in $X$. Since we only need to avoid the two vertices $u$ and $v$, there exists at least one neighbor of $b_2$ in $X$ that is neither $u$ nor $v$. Choose such a vertex as $b_2'' \in N_X(b_2) \cap X$. By Lemma \ref{P2C_Johnson}, there exist two vertex-disjoint paths $P_{11}$ and $P_{12}$ of $P2C(u, b_7; v, b_2'')$ that cover $X$. Adding the fault-free cross edges $ab_7$ and $b_2b_2''$, and concatenating $P_{11}$, $Q_1$, and $P_{12}$ yields a Hamiltonian path in $J(6, 3) - F$ joining $u$ and $v$.

\noindent\textbf{Case 2.} $u, v \in Y$.
Recall that $F_1 = F \setminus \{ab_2\}$ leaves exactly 3 faulty edges in $Y$. By Lemma \ref{FT_HC_J52}, there exists a Hamiltonian path $Q_2$ in $Y - F_1$ joining $u$ and $v$. However, $Q_2$ might contain the faulty edge $ab_2$. We discuss the following two subcases based on whether $ab_2$ is contained in $Q_2$.

\noindent\textbf{Case 2.1.} $ab_2 \in E(Q_2)$. We choose two distinct vertices $c_1, c_2 \in X$ such that $c_1 \in N_X(a)$ and $c_2 \in N_X(b_2)$. By Lemma \ref{HC_Johnson}, there exists a Hamiltonian path $P_2$ in $X - F$ joining $c_1$ and $c_2$. By deleting the edge $ab_2$, adding the edges $ac_1$ and $b_2c_2$, and concatenating these segments with $P_2$, we obtain a Hamiltonian path in $J(6, 3) - F$ joining $u$ and $v$.

\noindent\textbf{Case 2.2.} $ab_2 \notin E(Q_2)$. Then choose any edge $p_1p_2 \in E(Q_2)$ and two distinct vertices $c_1', c_2' \in X$ such that $c_1' \in N_X(p_1)$ and $c_2' \in N_X(p_2)$. By Lemma \ref{HC_Johnson}, there exists a Hamiltonian path $P_2'$ in $X - F$ joining $c_1'$ and $c_2'$. By deleting the edge $p_1p_2$, adding the cross edges $p_1c_1'$ and $p_2c_2'$, and concatenating the resulting segments of $Q_2$ with $P_2'$, we obtain a Hamiltonian path in $J(6, 3) - F$ joining $u$ and $v$.

\begin{table}[ht]
\centering
\renewcommand{\arraystretch}{1.2}
\begin{tabular}{c|c|c}
\hline
$b_2$ & $v$ & Hamiltonian path \\
\hline
125 & 123 & $\langle 125,126,256,236,136,356,156,135,235,123 \rangle$\\
125 & 126 & $\langle 125,135,136,236,256,156,356,235,123,126 \rangle$\\
125 & 135 & $\langle 125,126,123,235,256,236,136,356,156,135 \rangle$\\
125 & 136 & $\langle 125,126,123,235,135,156,256,236,356,136 \rangle$\\
125 & 156 & $\langle 125,135,136,236,126,123,235,256,356,156 \rangle$\\
125 & 235 & $\langle 125,135,136,236,256,356,156,126,123,235 \rangle$\\
125 & 236 & $\langle 125,135,136,126,123,235,256,156,356,236 \rangle$\\
125 & 256 & $\langle 125,135,136,236,126,123,235,356,156,256 \rangle$\\
125 & 356 & $\langle 125,126,123,235,135,136,156,256,236,356 \rangle$\\
\hline
\end{tabular}
\caption{Hamiltonian paths in $Y-F_Y$ joining $b_2=\{1,2,5\}$ and all other vertices in $Y$. For convenience, the set $\{i_1, i_2, i_3\}$ is abbreviated as $i_1i_2i_3$.}
\label{Table:FT_HC_J63_Special}
\end{table}

\noindent\textbf{Case 3.} $u \in X$ and $v \in Y$. Choose $v'\in\{b_2,b_3,b_4,b_5\}\setminus\{v\}$. By Table \ref{Table:FT_HC_J63_Special}, there exists a Hamiltonian path in $Y - F$ joining $b_2$ and $v$, denoted by $\langle b_2, Q_3, v \rangle$. For any $p \in N_X(b_2) \setminus \{u\}$, by Lemma \ref{HC_Johnson}, there exists a Hamiltonian path in $X - F$ joining $u$ and $p$, denoted by $\langle u, P_3, p \rangle$. Then the path $\langle u, P_3, p, b_2, Q_3, v \rangle$ is a Hamiltonian path in $J(6, 3) - F$ joining $u$ and $v$.
\end{proof}

\section{Edge-fault-tolerant Hamiltonian Connectivity of $J(n,k)$ under Matching Faults}

In the previous section, we showed that under unrestricted edge faults, the maximum number of faulty edges that can be guaranteed while preserving Hamiltonian connectivity of $J(n,k)$ is $k(n-k)-3$. This worst-case bound is determined by fault configurations highly concentrated around a single vertex. To investigate whether substantially larger fault sets can be tolerated when such concentration is excluded, we next consider matching faults, where no two faulty edges share a common endpoint. In this section, we prove that Hamiltonian connectivity is preserved after the deletion of an arbitrary matching.

We first consider the base cases for small networks.

\begin{lemma}\label{FT_M_J42}
For any two adjacent vertices $u$ and $v$ of $J(4,2)$ and any matching $F\subseteq E(J(4,2))$, the graph $J(4,2)-F$ contains a Hamiltonian path joining $u$ and $v$.
\end{lemma}

\begin{proof}
Since $u$ and $v$ are adjacent in $J(4,2)$, there exists an element $i\in u\setminus v$. Let $X=X(i)=\{x\in V(J(4,2)):i\in x\}$ and $Y=Y(i)=\{y\in V(J(4,2)):i\notin y\}$.
Then $X\cong K_3$ and $Y\cong K_3$. Write $V(X)=\{u_1,u_2,u_3\}$ and $V(Y)=\{v_1,v_2,v_3\}$, where the vertices are indexed so that $u_jv_j\notin E(J(4,2))$ for $j=1,2,3$. Hence, $E(X,Y)=\{u_jv_\ell: j\ne \ell\}$. Since $u\in X$, $v\in Y$, and $uv\in E(J(4,2))$, by relabeling the indices if necessary, we may assume that $u=u_2$, $v=v_1$. The decomposition of $J(4,2)$ into two copies of $K_3$, together with the corresponding cross edges, is illustrated in Fig. \ref{fig:J42}.

\noindent\textbf{Case 1.} $F\subseteq E(X,Y)$. In this case, all edges of $X$ and $Y$ are fault-free. Consider the two cross edges $u_1v_2$ and $u_3v_2$. Since they share the vertex $v_2$ and $F$ is a matching, at most one of them belongs to $F$. If $u_1v_2\notin F$, then $\langle u_2,u_3,u_1,v_2,v_3,v_1\rangle$ is a Hamiltonian path from $u$ to $v$  in $J(4,2)-F$. If $u_1v_2\in F$, then $u_3v_2\notin F$, and hence $\langle u_2,u_1,u_3,v_2,v_3,v_1\rangle$ is a Hamiltonian path from $u$ to $v$   in $J(4,2)-F$.

\begin{figure}
	\centering
	\includegraphics[width=0.35\textwidth]{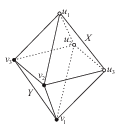}
     \caption{Structure of $J(4,2)$ under the partition $V(J(4,2))=X(i)\cup Y(i)$, where $X(i)\cong Y(i)\cong K_3$ and $u_jv_\ell\in E(J(4,2))$ if and only if $j\neq l$. }
     \label{fig:J42}
\end{figure}

\noindent\textbf{Case 2.} $F\cap(E(X)\cup E(Y))\ne\emptyset$.

By interchanging the roles of $X$ and $Y$ if necessary, we may assume that $F\cap E(X)\ne\emptyset$. Since $X\cong K_3$ and $F$ is a matching, $|F\cap E(X)|=1$. Extend $F$ to a maximal matching $F'$ of $J(4,2)$. Because the unique edge of $F\cap E(X)$ is incident with two of the three vertices of $X$, each of the other two edges of $E(X)$ shares an endpoint with it. Consequently, no additional edge of $E(X)$ can be added when extending $F$, and therefore $|F'\cap E(X)|=1$. Moreover, if a Hamiltonian path from $u$ to $v$ exists in $J(4,2)-F'$, then the same path also lies in $J(4,2)-F$, since $F\subseteq F'$. Thus, it suffices to consider maximal matchings $F'$ satisfying $|F'\cap E(X)|=1$.

Note that $J(4,2)$ is isomorphic to the octahedral graph. With the terminals fixed as $u=u_2$ and $v=v_1$, there are exactly nine maximal matchings satisfying $|F'\cap E(X)|=1$: six perfect matchings of size $3$ and three maximal matchings of size $2$. By the symmetry of $J(4,2)$ and the automorphisms preserving the terminal set $\{u_2, v_1\}$, the nine maximal matchings can be reduced to four essentially different configurations up to isomorphism. These four representative matchings, together with their corresponding fault-free Hamiltonian paths from $u_2$ to $v_1$, are listed in Table~\ref{Table:FT_M_J42}. Therefore, $J(4,2)-F'$ contains a Hamiltonian path from $u$ to $v$ in each case, and hence so does $J(4,2)-F$.

\end{proof}

\begin{table}[htbp]
    \centering
    \renewcommand{\arraystretch}{1.2}
    \begin{tabular}{c|c|c}
    \hline
        Size &  Faulty maximal matching $F$  &  Hamiltonian path \\
       \hline
       \multirow{2}{*}{3} 
       & $\{u_1u_3, u_2v_1, v_2v_3\}$ & $\langle u_2, v_3, u_1, v_2, u_3, v_1\rangle$\\
       & $\{u_1u_3, u_2v_3, v_1v_2\}$ & $\langle u_2, u_1, v_3, v_2, u_3, v_1\rangle$\\
%       & $\{u_2u_3, u_1v_3, v_1v_2\}$ & $\langle u_2, u_1, u_3, v_2, v_3, v_1\rangle$\\
%       & $\{u_2u_3, u_1v_2, v_1v_3\}$ & $\langle u_2, v_3, u_1, u_3, v_2, v_1\rangle$\\
%       & $\{u_1u_2, u_3v_1, v_2v_3\}$ & $\langle u_2, u_3, v_2, u_1, v_3, v_1\rangle$\\
%       & $\{u_1u_2, u_3v_2, v_1v_3\}$ & $\langle u_2, v_3, v_2, u_1, u_3, v_1 \rangle$\\
       \hline
       \multirow{2}{*}{2} 
       & $\{u_2u_3, v_2v_3\}$ & $\langle u_2, v_3, u_1, u_3, v_2, v_1\rangle$\\
       & $\{u_1u_2, v_1v_2\}$ & $\langle u_2, v_3, v_2, u_1, u_3, v_1\rangle$\\
        \hline
    \end{tabular}
    \caption{Hamiltonian paths in $J(4,2)-F$ joining $u_2$ and $v_1$ for all essentially different maximal matching faults.}
    \label{Table:FT_M_J42}
\end{table}

\begin{lemma}\label{FT_M_J52}
$J(5,2)$ is matching-fault-tolerant Hamiltonian-connected.
\end{lemma}

\begin{proof} 
Let $u$ and $v$ be two distinct vertices of $J(5,2)$, and let $F$ be a set of faulty edges that forms a matching in $J(5,2)$. Since $u \ne v$, there exists an element $i \in [5]$ such that $i \in u$ and $i \notin v$. Clearly, $X(i) \cong K_4$ and $Y(i) \cong J(4,2)$, with $u \in X(i)$ and $v \in Y(i)$.

Since $X(i) \cong K_4$ and $F$ is a matching, $X(i) - F$ always contains a cycle of length four. Consequently, there exists a Hamiltonian path in $X(i) - F$ starting at $u$, which we denote by $\langle u, u_1, u_2, u_3 \rangle$. 

Next, we need to find a fault-free edge from $u_3$ to a vertex $v_1 \in Y(i)$ such that $v_1$ is adjacent to $v$ in $Y(i)$. Let $u_3 = \{i, x\}$ for some element $x\in[5]$. Since $Y(i)$ consists of all 2-subsets of $[5] \setminus \{i\}$, the vertex $v$ can be written as $v = \{a, b\}$ where $a,b \neq i$. The neighbors of $u_3$ lying in $Y(i)$ consist of exactly three vertices, all of which contain $x$. We consider two cases based on whether $x \in v$.

If $x \in \{a,b\}$, say $x=a$, then $v = \{a,b\}$ is itself a neighbor of $u_3$ in $Y(i)$. The other two neighbors of $u_3$ in $Y(i)$ must be of the form $\{a, c\}$ and $\{a, d\}$ where $c, d \notin \{i, a, b\}$. So they are adjacent to $v$ in $Y(i)$.

If $x \notin \{a,b\}$, let $[5] \setminus \{i, x\} = \{a, b, c\}$. The three neighbors of $u_3$ in $Y(i)$ are $\{x, a\}$, $\{x, b\}$, and $\{x, c\}$. Notice that $\{x, a\}$ and $\{x, b\}$ each share exactly one element with $v = \{a,b\}$, meaning they are adjacent to $v$ in $Y(i)$.

In both cases, $u_3$ has at least two neighbors in $Y(i)$ that are adjacent to $v$. Since $F$ is a matching, at most one edge incident to $u_3$ belongs to $F$. Thus, there exists at least one fault-free edge from $u_3$ to a vertex $v_1 \in Y(i) \setminus \{v\}$ such that $v_1$ is adjacent to $v$. By Lemma \ref{FT_M_J42}, there exists a Hamiltonian path $P$ in $Y(i) - F$ joining $v_1$ and $v$. By adding the fault-free edge $u_3v_1$ and concatenating the path $\langle u, u_1, u_2, u_3 \rangle$ with $P$, we obtain a Hamiltonian path in $J(5,2) - F$ joining $u$ and $v$. 
\end{proof}

We are now ready to present the main result of this section. 

\begin{theorem}\label{FT_M_Johnson}
For $n\ge 5$ and $1\le k\le n-1$, $J(n,k)$ is matching-fault-tolerant Hamiltonian-connected.
\end{theorem}

\begin{proof}
We prove the theorem by induction on $n$.

For $k=1$, since $J(n,1)\cong K_n$, let $M = \{a_1b_1, a_2b_2, \cdots, a_tb_t\}$ be a matching in $K_n$, and let $U = V(J(n,1)) \setminus \bigcup_{i=1}^t \{a_i, b_i\}=\{u_1, u_2, \dots, u_s\}$ be the set of vertices not incident to any edge of $M$. By symmetry, we enumerate Hamiltonian paths of essentially different terminals in Table \ref{HC_M_K_n}. Since $J(n,n-1)\cong J(n,1)$, the result also holds for $k=n-1$. For $n=5$ and $k\in\{2,3\}$, the result follows from Lemma~\ref{FT_M_J52}.

Now assume that $n\ge6$ and $2\le k\le n-2$, and suppose that
$J(m,k')$ is matching-fault-tolerant Hamiltonian-connected for every
$5\le m<n$ and $1\le k'\le m-1$.

Since $J(n,k)\cong J(n,n-k)$, we may assume that $k\le n/2$. Let $F$ be any matching of $J(n,k)$, and let $u,v$ be any two distinct vertices of $J(n,k)-F$. Since $u\ne v$, there exists an element $i\in u\setminus v$. Let $X=X(i)$ and $Y=Y(i)$. Then $X\cong J(n-1,k-1), Y\cong J(n-1,k).$ Choose any vertex $a\in X\setminus\{u\}$. Since $a$ contains $i$, it has exactly $n-k$ neighbors in $Y$. Hence
\[
|N_Y(a)\setminus\{v\}|\ge n-k-1
\ge \frac n2-1\ge2.
\]
Since $F$ is a matching, at most one edge incident with $a$ belongs to $F$. Therefore, among the edges joining $a$ to $N_Y(a)\setminus\{v\}$, there exists a non-faulty edge $aa'$ for some $a'\in Y\setminus\{v\}$.

The restriction $F_X$ of $F$ to $E(X)$ is a matching, and similarly $F_Y$ is a matching. By the induction hypothesis, there exists a Hamiltonian path $\langle u,P_1,a\rangle$ in $X-F_X$, and a Hamiltonian path $\langle a',P_2,v\rangle$ in $Y-F_Y$. Since $aa'\notin F$, concatenating these paths with the edge $aa'$ gives $\langle u,P_1,a,a',P_2,v\rangle$, which is a Hamiltonian $u$-$v$ path of $J(n,k)-F$. Thus $J(n,k)$ is matching-fault-tolerant Hamiltonian-connected.
\end{proof}

\begin{table}[htbp]
    \centering
    \renewcommand{\arraystretch}{1.2}
    \begin{tabular}{c|c|l}
    \hline
     &  Terminal pair $\{u, v\}$    &  Hamiltonian path in $J(n,1)-M$ \\
      \hline
      1  & $\{a_1,b_1\}$  & $\langle a_1,\dots,a_t,u_1,\dots,u_s,b_t,\dots,b_1\rangle $\\
      2  & $\{a_1,a_2\}$ & $\langle a_1,a_3,\dots,a_t,u_1,\dots,u_s,b_t,\dots,b_1,a_2 \rangle$ \\
      3  & $\{a_1,u_1\}$ & $\langle a_1,\dots,a_t,u_2,\dots,u_s,b_t,\dots,b_1,u_1 \rangle$ \\
      4  & $\{u_1,u_2\}$ & $\langle u_1,a_1,\dots,a_t,u_3,\dots,u_s,b_t,\dots,b_1,u_2 \rangle$\\
       \hline
    \end{tabular}
    \caption{Hamiltonian paths for all representative terminal pairs in $K_n$ under a matching fault.}
    \label{HC_M_K_n}
\end{table}

As an immediate consequence, the result also applies to the maximum-cardinality case in which the faulty set is a perfect matching.

\begin{corollary}
Let $F$ be a perfect matching of $J(n,k)$, where $n\ge5$ and $1\le k\le n-1$. Then $J(n,k)-F$ is Hamiltonian-connected.
\end{corollary}

\section{Vertex-fault-tolerant Hamiltonian connectivity of $J(n,k)$}

In this section, we investigate the Hamiltonian connectivity of Johnson graphs under vertex faults. Unlike an edge fault, which removes a single communication link, a vertex fault removes a
vertex together with all of its incident edges and may therefore cause more extensive structural disruption.

For example, an exhaustive verification shows that
the degree-based edge-fault bound cannot in general be carried over to vertex faults. In $J(6,3)$, deleting the six vertices
\[
\{1,2,3\},\ \{1,3,4\},\ \{1,4,6\},\ \{3,4,5\},\
\{3,5,6\},\ \{4,5,6\}
\]
yields a subgraph containing no Hamiltonian path from $\{2,3,5\}$ to $\{2,4,5\}$. We therefore seek a vertex-fault tolerance guarantee for Johnson graphs in this section.

The following lemma characterizes the intersection of the cross-layer neighborhoods of two vertices in $X$.

\begin{lemma}\label{Common_neighbor}
For any two distinct vertices $a,b\in X$, we have $|N_Y(a)|=|N_Y(b)|=n-k$ and 
\[ |N_Y(a)\cap N_Y(b)|=
    \begin{cases}
        1,& \text{if }a\text{ is adjacent to } b;\\
        0, & \text{otherwise}.
    \end{cases}
\]
\end{lemma}
\begin{proof}
Since $a$ and $b$ are vertices in $X$, they both contain the element $i$. By definition, replacing the element $i$ in $a$ (or $b$) with any element not currently in the set yields a neighbor in $Y$. Thus, $|N_Y(a)|=|N_Y(b)|=n-k$. We now determine the size of their common neighborhood based on their adjacency in $X$.

\noindent\textbf{Case 1.} $a$ and $b$ are adjacent. Without loss of generality, let $a=\{s_1,i,t_3,\cdots,t_k\}$ and $b=\{s_2,i,t_3,\cdots,t_k\}$, where $s_1\neq s_2$. To find a common neighbor in $Y$, we must replace the element $i$ in both sets such that the resulting sets are identical. The only way to achieve this is to replace $i$ in $a$ with $s_2$, and replace $i$ in $b$ with $s_1$. Thus, their unique common neighbor in $Y$ is exactly $\{s_1, s_2, t_3, \dots, t_k\}$. Hence, $|N_Y(a) \cap N_Y(b)| = 1$.

\noindent \textbf{Case 2.} $a$ and $b$ are non-adjacent. Without loss of generality, let $a=\{s_1, p_1, i, \dots, t_k\}$ and $b=\{s_2, p_2, i, \dots, t_k\}$, where $\{s_1, p_1\} \cap \{s_2, p_2\} = \emptyset$. For any vertex $v \in N_Y(a)$, $v$ must not contain $i$ (since $v \in Y$), which means $i$ must be replaced by some element  not in $a$. Consequently, $v$ will still contain both $s_1$ and $p_1$. However, $b$ lacks both $s_1$ and $p_1$. A neighbor of $b$ in $Y$ can acquire at most one of them (by replacing $i$). Therefore, no vertex in $N_Y(b)$ can contain both $s_1$ and $p_1$, which implies $v \notin N_Y(b)$. Thus, $N_Y(a) \cap N_Y(b) = \emptyset$.
\end{proof}

For the small base cases, we use exhaustive computational verification, which yields the following lemma.

\begin{lemma}\label{FT_V_J52}
$J(5,2)$ and $J(5,3)$ are $3$-vertex-fault-tolerant Hamiltonian-connected.
\end{lemma}

With the foundational base cases established, we now present the main theorem regarding the vertex-fault-tolerant Hamiltonian connectivity of $J(n,k)$.

\begin{theorem}\label{FT_V_Johnson}
$J(n,k)$ is $(n-2)$-vertex-fault-tolerant Hamiltonian-connected whenever $n\ge 5$ and $1\le k \le n-1$.
\end{theorem}

\begin{proof} 
We proceed by induction on $n$. Since $J(n,1) \cong J(n,n-1) \cong K_n$, the complete graph is trivially $(n-2)$-vertex-fault-tolerant Hamiltonian-connected. For the base cases $n=5$, $J(5,2)$ and $J(5,3)$ are $3$-vertex-fault-tolerant Hamiltonian-connected by Lemma \ref{FT_V_J52}. Since $J(n,k)\cong J(n,n-k)$, we may assume that $k\le n/2$.
 
For the inductive step, assume that $n\ge 6$ and $2\le k \le n-2$. The strong induction hypothesis asserts that $J(m,k')$ is $(m-2)$-vertex-fault-tolerant Hamiltonian-connected for any $5 \le m < n$ and $1 \le k' \le m-1$. Let $F$ be a set of faulty vertices in $J(n,k)$ with $|F| \le n-2$.

Consider two distinct arbitrary vertices $u$ and $v$ in $J(n,k)-F$. Since $u \ne v$, there exists an element $i\in [n]$ such that $i\in u$ and $i\notin v$. We partition $J(n,k)$ along $i$, yielding $X(i) \cong J(n-1,k-1)$ and $Y(i) \cong J(n-1,k)$, with $u \in X(i)$ and $v \in Y(i)$. Let $F_X$ and $F_Y$ denote the restriction of $F$ on $X(i)$ and $Y(i)$, respectively. We consider two subcases based on the fault distribution.

\noindent\textbf{Case 1.} $1\le|F_X|\le n-3$ and $1\le|F_Y|\le n-3$. 
We first select candidate endpoints in $X(i)$ to cross into $Y(i)$. The number of available fault-free vertices in $X(i)$ other than $u$ is $|X(i)| - |F_X| - 1$. 
If we can choose two distinct candidates $a, b \in X(i) \setminus (F_X \cup \{u\})$, then by Lemma 13, the cardinality of the neighborhood of $\{a,b\}$ in $Y(i)$ satisfies:
\begin{align*}
    |N_Y(\{a,b\})| &\ge 2(n-k)-1 \\
    &= (n-k) + (n-k) - 1 \\
    &\ge (n-k) + k - 1 \quad (\text{since } n \ge 2k) \\
    &= n-1.
\end{align*}
Since $|F_Y \cup \{v\}| \le |F_Y| + 1 \le (n-3) + 1 = n-2$, we have $|N_Y(\{a,b\})| \ge n-1 > |F_Y \cup \{v\}|$. Thus, there exists at least one fault-free cross edge to $Y(i)$.
Note that we can always find two such candidates $a,b$ unless $k=2$ and $|F_X| = n-3$. In this extreme scenario, only one candidate $a$ remains in $X(i)$, but consequently $|F_Y| \le (n-2) - (n-3) = 1$. The single candidate $a$ has $|N_Y(a)| = n-2 \ge 4$ neighbors in $Y(i)$, which is greater than $|F_Y \cup \{v\}| \le 2$, ensuring a fault-free cross edge.

Without loss of generality, let $c \in Y(i) \setminus (F_Y \cup \{v\})$ be a fault-free neighbor of $b$. By the induction hypothesis, there exists a Hamiltonian path $\langle u, P_1, b \rangle$ in $X(i)-F_X$, and a Hamiltonian path $\langle c, P_2, v \rangle$ in $Y(i)-F_Y$. Concatenating them via the edge $bc$ yields a Hamiltonian path $\langle u, P_1, b, c, P_2, v \rangle$ in $J(n,k)-F$.

\noindent\textbf{Case 2.} $|F_X|=0$ or $n-2$. We only consider $F_X = \emptyset$ as the case $F_Y = \emptyset$ can be handled symmetrically. Let $f$ be a faulty vertex in $F_Y$. Then, the remaining faults satisfy $|F_Y \setminus \{f\}| \le n-3$. 

By the induction hypothesis, there exists a Hamiltonian path from $f$ to $v$ in $Y(i) - (F_Y \setminus \{f\})$, denoted by $\langle f, f', Q_2, v \rangle$, where $f'$ is the immediate neighbor of $f$. Since $f$ is faulty, we truncate it to obtain a fault-free path segment $\langle f', Q_2, v \rangle$ in $Y(i)-F_Y$. 
Because $k \ge 2$, the vertex $f'$ has $k \ge 2$ neighbors in $X(i)$. Thus, we can choose a neighbor $f'' \in N_{X(i)}(f')$ such that $f'' \ne u$. By the induction hypothesis on $X(i)$, there exists a Hamiltonian path $\langle u, Q_1, f'' \rangle$ in $X(i)$. Connecting the two segments via the cross edge $f''f'$, the path $\langle u, Q_1, f'', f', Q_2, v \rangle$ is a fault-free Hamiltonian path joining $u$ and $v$ in $J(n,k)-F$. 
\end{proof}

\section{Fault-Tolerant Algorithm and Complexity Analysis}

In this section, we establish efficient algorithms based on the constructive proofs presented previously. The constructive logic for handling general edge faults, matching faults, and vertex faults shares a fundamental divide-and-conquer architecture based on the subgraph partitioning of $J(n,k)$. We first present the framework for general edge faults, which incorporates a dynamic partition selection subroutine. Subsequently, we introduce a unified routing algorithm for matching and vertex faults. Finally, we rigorously analyze the complexity of the proposed algorithms.

\subsection{Routing under General Edge Faults}
Throughout this section, every recursive instance is normalized using the complement isomorphism
\[
J(n,k)\cong J(n,n-k).
\]
Specifically, whenever $k>n/2$, each vertex is replaced by its complement in $[n]$, the parameter $k$ is replaced by $n-k$, and the terminals and fault set are transformed accordingly. The
routing algorithm is then applied to the normalized instance, and the resulting path is mapped back by the inverse complement transformation. Hence, every nontrivial recursive instance may be assumed to satisfy $k\le n/2$.

When the recursion reaches $k=1$ (or, equivalently, $k=n-1$), the problem reduces to
Hamiltonian routing in a faulty complete graph. Let
$G=K_n-F$, where $|F|\le n-4$. For every pair of nonadjacent
vertices $x,y$ in $G$, we have $xy\in F$ and
\[
d_F(x)+d_F(y)\le |F|+1\le n-3.
\]
Consequently,
\[
d_G(x)+d_G(y)
=2(n-1)-d_F(x)-d_F(y)
\ge n+1.
\]
Thus, $G$ satisfies the Ore-type sufficient condition for
Hamiltonian connectivity.

We use \texttt{CompleteGraphRouting} to denote the standard constructive procedure associated with this Ore-type condition. The faulty edges are processed successively. Whenever a faulty
edge occurs on the current Hamiltonian path, the usual Ore-type path-rearrangement operation is applied to eliminate it while preserving the prescribed endpoints. Since at most $n-4$ faulty
edges are processed and each path rearrangement requires at most a linear scan of the current path, a single invocation takes $O(n^2)$ time.

Since $J(4,2)$ and $J(5,2)$ have small orders, all admissible fault configurations and terminal pairs can be precomputed exhaustively. We therefore treat \texttt{SmallGraphRouting} as a
constant-time lookup routine.

Since $J(6,3)$ has constant size, all exceptional configurations covered by Lemma~\ref{FT_HC_J63_Special} can be handled by a
finite lookup table after canonical relabeling. We denote this constant-time procedure by \texttt{SpecialJ63Routing} in Algorithm \ref{alg:routing_general}. This approach guarantees that the base-case anomalies are resolved in $O(1)$ time.

\begin{algorithm}[!htbp]
\caption{Hamiltonian Routing in $J(n,k)$ with General Edge Faults}
\label{alg:routing_general}

\SetKwProg{Fn}{Function}{:}{}
\SetKwFunction{FMainEdge}{RoutingGenEdge}
\SetKwFunction{FSelect}{PartitionSelection}

\KwIn{Graph parameters $n, k$; Terminals $u, v$; Fault set $F$ with $|F| \le k(n-k)-3$}
\KwOut{A fault-free Hamiltonian path $P$ joining $u$ and $v$}
\BlankLine

\Fn{\FMainEdge{$n, k, u, v, F$}}{
   \If{$k > n/2$}{
        Transform $(u,v,F)$ to the complementary instance
        $(\bar u,\bar v,\bar F)$ of $J(n,n-k)$\;
        $\bar P \leftarrow$
        \FMainEdge{$n,n-k,\bar u,\bar v,\bar F$}\;
        \Return the inverse-complement image of $\bar P$\;
    }
    \lIf{$k = 1$}{
        \Return CompleteGraphRouting($n, u, v, F$)
    }
    \If{$(n,k)=(4,2)$ \textbf{or} $(n,k)=(5,2)$}{
        \Return SmallGraphRouting($n,k,u,v,F$)\;
    }
    \If{The exceptional configuration of $J(6,3)$ in Lemma~\ref{FT_HC_J63_Special} occurs}{
        \Return SpecialJ63Routing($u, v, F$)\;
    }

    $i \leftarrow$ \FSelect{$n, k, F$}\;
    Partition into $X(i) \cong J(n-1, k-1)$ and $Y(i) \cong J(n-1, k)$\;
    
    \uIf{$u \in X(i)$ \textbf{and} $v \in Y(i)$}{
        Select $a \in X(i)$ and $a' \in Y(i)$ such that $aa' \notin F$\;
        $P_1 \leftarrow $ \FMainEdge{$n-1, k-1, u, a, F_X$}\;
        $P_2 \leftarrow $ \FMainEdge{$n-1, k, a', v, F_Y$}\;
        \Return $P_1 \circ aa' \circ P_2$\;
    }
    \uElseIf{$u \in Y(i)$ \textbf{and} $v \in X(i)$}{
        Select $a \in Y(i)$ and $a' \in X(i)$ such that $aa' \notin F$\;
        $P_1 \leftarrow $ \FMainEdge{$n-1, k, u, a, F_Y$}\;
        $P_2 \leftarrow $ \FMainEdge{$n-1, k-1, a', v, F_X$}\;
        \Return $P_1 \circ aa' \circ P_2$\;
    }
    \uElseIf{$u \in X(i)$ \textbf{and} $v \in X(i)$}{
        $P_X \leftarrow $ \FMainEdge{$n-1, k-1, u, v, F_X$}\;
        Scan the edges of $P_X$ until an edge $x_1x_2$ admitting two distinct fault-free cross-edges $x_1y_1$ and $x_2y_2$ is found\;
        Split $P_X$ at $(x_1, x_2)$ into $P_X^1$ (from $u$ to $x_1$) and $P_X^2$ (from $x_2$ to $v$)\;
        $P_Y \leftarrow $ \FMainEdge{$n-1, k, y_1, y_2, F_Y$}\;
        \Return $P_X^1 \circ x_1 y_1 \circ P_Y \circ y_2 x_2 \circ P_X^2$\;
    }
    \ElseIf{$u \in Y(i)$ \textbf{and} $v \in Y(i)$}{
        $P_Y \leftarrow $ \FMainEdge{$n-1, k, u, v, F_Y$}\;
        Select an edge $(y_1, y_2) \in E(P_Y)$ and two cross-edges $y_1 x_1, y_2 x_2 \notin F$\;
        Split $P_Y$ at $(y_1, y_2)$ into $P_Y^1$ (from $u$ to $y_1$) and $P_Y^2$ (from $y_2$ to $v$)\;
        $P_X \leftarrow $ \FMainEdge{$n-1, k-1, x_1, x_2, F_X$}\;
        \Return $P_Y^1 \circ y_1 x_1 \circ P_X \circ x_2 y_2 \circ P_Y^2$\;
    }
}
\end{algorithm}

\begin{algorithm}[!t]
\caption{PartitionSelection for General Edge Faults}
\label{alg:partition_selection}

\SetKwProg{Fn}{Function}{:}{}
\SetKwFunction{FSelect}{PartitionSelection}

\KwIn{Parameters $n, k$; Faulty-edge set $F$}
\KwOut{An element $i \in [n]$ satisfying fault bounds for $X(i)$ and $Y(i)$}
\BlankLine

\Fn{\FSelect{$n, k, F$}}{
    \ForEach{$i \in [n]$}{
        $C_X[i] \leftarrow 0$\;
        $C_Y[i] \leftarrow 0$\;
    }
    
    \ForEach{edge $e = (x, y) \in F$}{
        
        \ForEach{$i \in x \cap y$}{
            $C_X[i] \leftarrow C_X[i] + 1$\;
        }
        
        \ForEach{$i \in [n] \setminus (x \cup y)$}{
            $C_Y[i] \leftarrow C_Y[i] + 1$\;
        }
    }
    
    \tcp{Select the valid partition element}
    \ForEach{$i \in [n]$}{
        \If{$C_X[i] \le (n-k)(k-1)-3$ \textbf{and} $C_Y[i] \le (n-k-1)k-3$}{
            \Return $i$\;
        }
    }
    \Return Error \tcp*{Unreachable by Theorem~\ref{FT_HC_Johnson}}
}
\end{algorithm}

The structural correctness and termination of Algorithm \ref{alg:routing_general} follow directly from the constructive proof of Theorem \ref{FT_HC_Johnson}. We now analyze its complexity. Let $N = \binom{n}{k}$ denote the total number of vertices in $J(n,k)$. 

\begin{theorem}
For $4\le k\le n-4$, Algorithm~\ref{alg:routing_general}
constructs a fault-free Hamiltonian path in asymptotically optimal
$\Theta(N)$ time, where $N=\binom{n}{k}$.
\end{theorem}

\begin{proof}
The total running time consists of path construction, partition
selection, bridge selection, and the base-case routing procedures.

The recursive construction outputs each vertex exactly once. With
linked path representations, concatenation of path segments requires
constant time. Hence the pure path-construction cost is $O(N)$.

Let $D$ be the maximum recursion depth. At each recursive step, the
first graph parameter decreases from $n$ to $n-1$, and therefore
\[
D=O(n).
\]

For a fixed recursion depth $d$, let $\mathcal{R}_d$ be the set of
recursive subproblems at that depth. The corresponding subgraphs are
vertex-disjoint, and hence their internal faulty-edge sets are
disjoint. Therefore,
\[
\sum_{R\in\mathcal{R}_d}|F_R|\le |F|.
\]

For each faulty edge $xy\in F_R$, Algorithm
\ref{alg:partition_selection} (updating the counters) needs $O(n)$ time. Thus the aggregate
partition-selection cost at one recursion depth is
\[
O\left(
n\sum_{R\in\mathcal{R}_d}|F_R|
\right)
=
O(n|F|).
\]
Summing over all $O(n)$ recursion levels gives
\[
T_{\mathrm{partition}}=O(n^2|F|).
\]
Since
\[
|F|\le k(n-k)-3\le \frac{n^2}{4},
\]
we obtain
\[
T_{\mathrm{partition}}=O(n^4).
\]
For $4\le k\le n-4$,
\[
N=\binom{n}{k}\ge\binom{n}{4}=\Theta(n^4),
\]
and hence
\[
T_{\mathrm{partition}}=O(N).
\]

The search for the required fault-free cross-layer edges is carried out along the recursively constructed paths, and its cost is included
in the overall path-processing overhead.

It remains to consider the base cases. For a complete-graph leaf $R\cong K_{m_R}$, let $F_R$ denote its internal faulty-edge set. A call to \texttt{CompleteGraphRouting} on $R$ requires
$O(m_R(|F_R|+1))$ time. Since these complete-graph leaves are vertex-disjoint,
\[
\sum_R m_R\le N,
\qquad
\sum_R |F_R|\le |F|.
\]
Moreover, $m_R\le n$ for every complete-graph leaf. Hence,
\[
T_{\mathrm{complete}}
=
O(N+n|F|)
=
O(N+n^3)
=
O(N).
\]

The fixed-size routines \texttt{SmallGraphRouting} and \texttt{SpecialJ63Routing} contribute only $O(N)$ aggregate time. Combining all components,
\[
T(N)=O(N).
\]
Since explicitly outputting a Hamiltonian path containing all $N$
vertices requires $\Omega(N)$ time, we conclude that
\[
T(N)=\Theta(N).
\]
Thus, Algorithm~\ref{alg:routing_general} is asymptotically optimal.
\end{proof}

The boundary cases $k=2$ and $3$ require separate consideration. For both cases, the admissible number of general edge faults is
linear in $n$. Since the recursion depth is $O(n)$, the above partition-selection analysis gives an aggregate overhead of $O(n^3)$. For $k=3$, we have
\[
N=\binom{n}{3}=\Theta(n^3),
\]
and hence the partition-selection overhead remains $O(N)$. For $k=2$, however,
\[
N=\binom{n}{2}=\Theta(n^2),
\]
so the same analysis yields the conservative bound
\[
O(n^3)=O(N^{3/2}).
\]
Thus, under the direct implementation of
Algorithm~\ref{alg:partition_selection}, the present analysis does not establish a linear-in-$N$ worst-case bound for $k=2$. A more specialized or incrementally maintained partition-selection scheme may further reduce this overhead.

\subsection{Routing under Matching and Vertex Faults}

\begin{algorithm}[!t]
\caption{Fault-Tolerant Hamiltonian Path for Matching/Vertex Faults}
\label{alg:routing_match_vertex}

\SetKwProg{Fn}{Function}{:}{}
\SetKwFunction{FMain}{RoutingMatchVertex}
\SetKwFunction{FBridge}{FindBridge}

\KwIn{Graph parameters $n, k$; Terminals $u, v$; Fault set $F$; $\mathrm{FaultType} \in \{\mathrm{Matching}, \mathrm{Vertex}\}$}
\KwOut{A fault-free Hamiltonian path $P$ joining $u$ and $v$ in $J(n,k) - F$}
\BlankLine

\Fn{\FMain{$n, k, u, v, F, \mathrm{FaultType}$}}{
    \If{$k>n/2$}{
    Transform $(u,v,F)$ into the complementary instance
    $(\bar{u},\bar{v},\bar{F})$ of $J(n,n-k)$\;
    
    $\bar{P} \leftarrow
    \FMain{$n,n-k,\bar{u},\bar{v},\bar{F},\mathrm{FaultType}$}$\;
    
    \Return the inverse-complement image of $\bar{P}$\;
    }

    \If{$k = 1$}{
        \Return CompleteGraphRouting($n, u, v, F, \mathrm{FaultType}$)\;
    }
    \If{$n = 5$}{
        \Return LookupBaseSmallGraphs($k, u, v, F, \mathrm{FaultType}$)\;
    }

    Select an arbitrary element $i \in u \setminus v$\;
    Partition into $X(i) \cong J(n-1, k-1)$ and $Y(i) \cong J(n-1, k)$ with $u \in X(i), v \in Y(i)$\;
    $F_X \leftarrow F \cap X(i)$, \quad $F_Y \leftarrow F \cap Y(i)$\;
    
    $(a, a') \leftarrow $ \FBridge{$X(i), Y(i), u, v, F_X, F_Y, \mathrm{FaultType}, F$}\;
    
    $P_1 \leftarrow $ \FMain{$n-1, k-1, u, a, F_X, \mathrm{FaultType}$}\;
    $P_2 \leftarrow $ \FMain{$n-1, k, a', v, F_Y, \mathrm{FaultType}$}\;
    
    \Return $P_1 \circ aa' \circ P_2$\;
}
\end{algorithm}

\begin{algorithm}[!t]
\caption{Bridge Selection for Matching and Vertex Faults}
\label{alg:find_bridge_mv}

\SetKwProg{Fn}{Function}{:}{}
\SetKwFunction{FBridge}{FindBridge}

\KwIn{Subgraphs $X, Y$; Terminals $u, v$; Fault sets $F_X, F_Y$; $\mathrm{FaultType}$; Global fault set $F$}
\KwOut{A valid fault-free cross-edge $(a, c)$ bridging $X$ and $Y$}
\BlankLine

\Fn{\FBridge{$X, Y, u, v, F_X, F_Y, \mathrm{FaultType}, F$}}{
    \uIf{$\mathrm{FaultType} == \mathrm{Matching}$}{
        Select an arbitrary candidate $a \in X \setminus \{u\}$\;
        \ForEach{$c \in N_Y(a) \setminus \{v\}$}{
            \If{$(a, c) \notin F$}{
                \Return $(a, c)$\;
            }
        }
    }
    \ElseIf{$\mathrm{FaultType} == \mathrm{Vertex}$}{
        \uIf{$1 \le |F_X| \le n-3$ \textbf{and} $1 \le |F_Y| \le n-3$}{
            Select a set $C \subseteq X \setminus (F_X \cup \{u\})$ with size $\min(2, |X \setminus (F_X \cup \{u\})|)$\;
            \ForEach{$a \in C$}{
                Select two distinct vertices $c_1,c_2\in N_Y(a)\setminus\{v\}$\;
                \lIf{$ac_1\notin F$}{
                \Return $(a,c_1)$}
                \Else{
                \Return $(a,c_2)$
                }
            }
        }
        \Else{
            \lIf{$F_Y \neq \emptyset$}{$f \leftarrow \text{any node in } F_Y$}
            \lElse{$f \leftarrow \text{any node in } Y \setminus \{v\}$}
            
            Select a fault-free neighbor $f' \in N_Y(f) \setminus F_Y$\;
            Select a fault-free neighbor $f'' \in N_X(f') \setminus \{u\}$\;
            \Return $(f'', f')$ \tcp*{Return a fault-free bridge}
        }
    }
}
\end{algorithm}

We now establish the complexity of the unified routing algorithm for matching and vertex faults. To achieve asymptotic optimality, we assume that the fault set $F$ is pre-processed into a hash-based data structure, enabling $\mathcal{O}(1)$ average-time complexity for fault-membership queries (i.e., verifying if a vertex or edge is faulty).

\begin{theorem} \label{thm:complexity_mv}
Algorithm \ref{alg:routing_match_vertex} constructs a fault-free Hamiltonian path in asymptotically optimal $\Theta(N)$ time.
\end{theorem}

\begin{proof}
The total time complexity of the algorithm is determined by the recursive combinatorial path generation and the local overhead associated with the \texttt{FindBridge} subroutine at each recursion step. 

Let $T(n,k)$ denote the time required to construct the Hamiltonian path in $J(n,k)$. Thus, the recurrence relation can be formulated as:
\[
    T(n,k) = T(n-1, k-1) + T(n-1, k) + T_{\mathrm{bridge}}(n,k) + \mathcal{O}(1),
\]
where $T_{\mathrm{bridge}}(n,k)$ is the time complexity of selecting a valid fault-free cross-edge.

Instead of exhaustively scanning the $\mathcal{O}(n)$ cross-neighbors for a selected candidate $a \in X(i)$, we exploit the structural constraints of the fault models to achieve $\mathcal{O}(1)$ bridge selection:
\begin{itemize}
    \item \textbf{Under matching faults:} Since the fault set is a matching, at most one faulty edge is incident to any vertex $a$. Therefore, the algorithm only generates at most two candidate cross-neighbors via elementary set operations (e.g., $a \setminus \{j\} \cup \{i\}$) and test them using $\mathcal{O}(1)$ queries. 
    \item \textbf{Under vertex faults:} Whether selecting an arbitrary pair $a, b$ (when faults are distributed) or arbitrarily picking an fault node $f$ in $F_X$ (when $F_X=\emptyset$), the initial selections require no heuristic search and execute in strict $\mathcal{O}(1)$ time. To locate the required intermediate fault-free neighbor, the algorithm utilizes $\mathcal{O}(1)$ hash queries in total.      
\end{itemize}

Consequently, $T_{\mathrm{bridge}}(n,k) = \mathcal{O}(1)$. The recurrence relation thus simplifies to:
\[
    T(n,k) \le T(n-1, k-1) + T(n-1, k) + \mathcal{O}(1).
\]
Notice that the size of the graph precisely follows Pascal's identity:
\[
    N(n,k) = \binom{n}{k} = \binom{n-1}{k-1} + \binom{n-1}{k}.
\]
By induction, solving this recurrence yields $T(n,k) = \mathcal{O}(N(n,k)) = \mathcal{O}(N)$. 

Furthermore, any algorithm that constructs and outputs a Hamiltonian path must sequentially visit and output all $N$ vertices, which requires at least $\Omega(N)$ time. Combining the upper and lower bounds, we conclude that the total time complexity is $T(n,k) = \Theta(N)$, proving that the algorithm is asymptotically optimal.
\end{proof}

\subsection{Simulation and Performance Evaluation}

To evaluate the computational efficiency of the proposed routing algorithms, we measure the Average Execution Time (AET) required to construct a fault-free Hamiltonian path. The experiments are conducted on $J(16,k)$ for $2\leq k\leq 8$. For each network size and fault model, one fault configuration at the maximum cardinality guaranteed by the corresponding theoretical result is generated and fixed across 100 randomly selected source-destination pairs. The timing interval includes only the execution of the routing algorithm and excludes graph construction, fault generation, terminal-pair sampling, and post-routing correctness verification. All experiments were implemented in Python 3.7 and executed on a machine equipped with an Intel(R) Core(TM) i7-8550U processor and 8 GB of RAM running Windows 10. All reported routing times were measured using a single thread.

Fig. \ref{AET_FT} reports the AET of the three fault-tolerant Hamiltonian routing algorithms for $J(16,k)$, $2\leq k\leq 8$, corresponding to network sizes ranging from $N=120$ to $N=12{,}870$. The fault configurations contain $k(n-k)-3$ general faulty edges,
a faulty-edge set forming a perfect matching, and $n-2$ faulty vertices for the general-edge, matching, and vertex fault models, respectively. For every tested network size and fault model, all 100 routing trials successfully produced valid fault-free Hamiltonian paths.

As shown in Fig. \ref{AET_FT}, the execution times of all three algorithms increase smoothly with the network size and exhibit an approximately linear growth trend over the tested range. As $N$ increases from 120 to $12{,}870$, the AET increases from $2.76\times10^{-3}$ to $3.00\times10^{-1}$~s under general edge faults, from $2.49\times10^{-3}$ to $3.18\times10^{-1}$~s under matching faults, and from $1.83\times10^{-3}$ to $2.74\times10^{-1}$~s under vertex faults. Linear regressions of AET against $N$ yield coefficients of determination of $R^2=0.9990$, $0.9978$, and $0.9982$ for the general-edge, matching, and vertex fault models, respectively. These high goodness-of-fit values indicate that the measured runtime
is well approximated by a linear function of $N$ over the tested range, supporting the theoretical scalability analysis developed in the preceding subsection.

The three curves remain relatively close to one another across all tested network scales, indicating comparable practical efficiency under the different fault models. The vertex-fault algorithm consistently exhibits the lowest AET, reaching approximately $0.274$~s for $J(16,8)$, whereas the general-edge and matching algorithms require approximately $0.300$ and $0.318$~s, respectively. Despite these modest implementation-level differences, no qualitatively different growth behavior is observed among the three fault models. It is also noteworthy that the matching-fault experiment for $J(16,8)$ contains a perfect matching of $6{,}435$ faulty edges, compared with only $61$ arbitrary faulty edges under the general-edge model. Thus, although the matching-fault model permits a substantially larger number of faulty edges because of its structural restriction, its routing time remains comparable to that under general edge faults.

Since the algorithms were implemented in Python, the absolute execution times
are implementation dependent; nevertheless, the observed growth trend agrees well with the theoretical scalability analysis.

\begin{figure}[ht]
	\centering
	\includegraphics[width=0.8\textwidth]{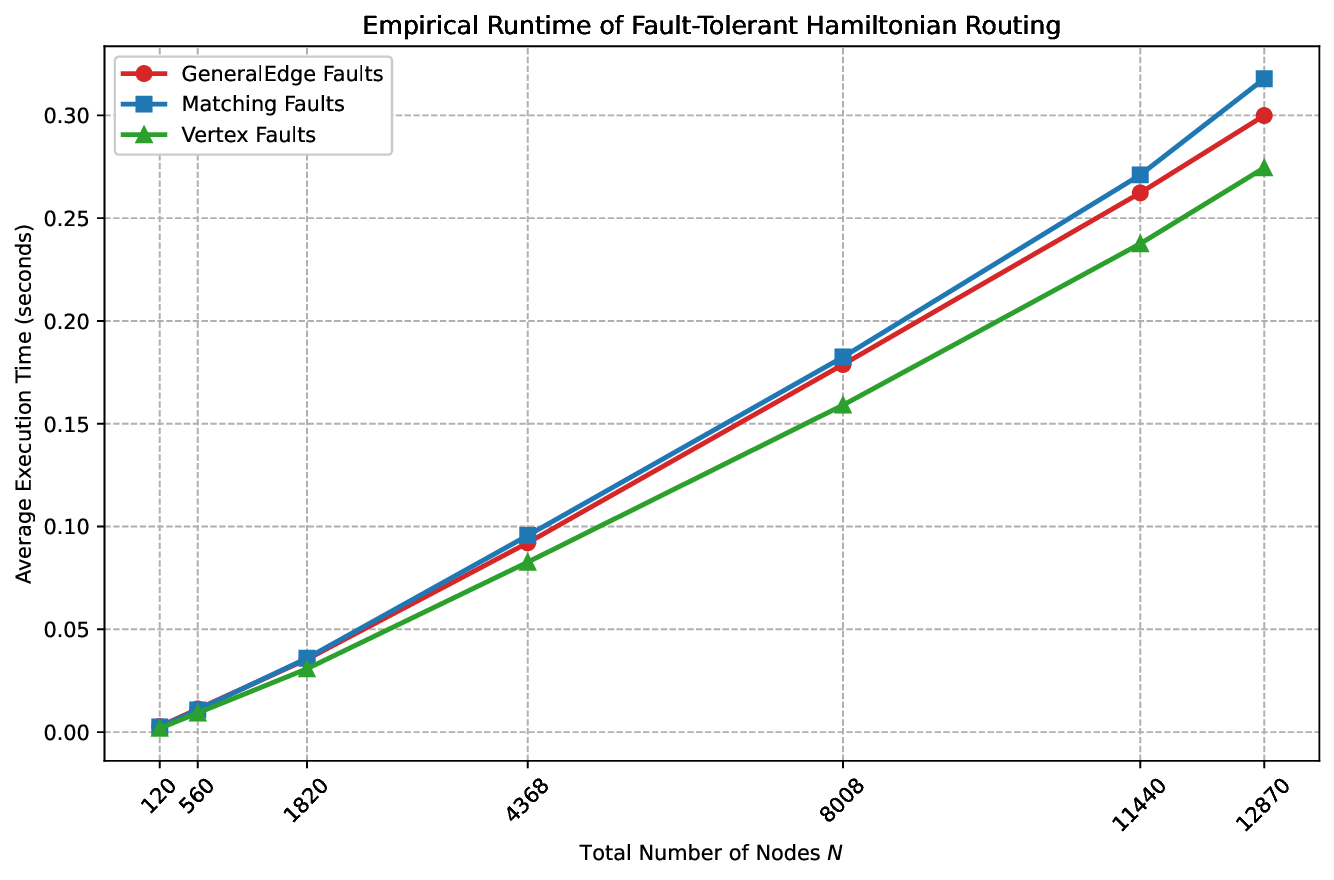}
     \caption{Average execution time (AET) versus network size $N$ for the proposed Hamiltonian routing algorithms under general edge, matching, and vertex faults in $J(16,k)$, $2\leq k\leq8$. Each data point is averaged over 100 randomly selected source-destination pairs at the maximum fault cardinality guaranteed by the corresponding fault-tolerance result.}
     \label{AET_FT}
\end{figure}

\section{Conclusion}

In this paper, we have investigated the fault-tolerant Hamiltonian connectivity of the Johnson graph $J(n,k)$, a highly symmetric graph with potential as a topology for multiprocessor systems. We first established that $J(n,k)$ is $(k(n-k)-3)$-edge-fault-tolerant Hamiltonian-connected, which attains the natural degree-based upper bound. Second, to reflect practical engineering scenarios where faulty links do not cluster, we proved that $J(n,k)$ remains Hamiltonian-connected under any matching fault for $n \ge 5$. Finally, addressing the destructive nature of node failures, we demonstrated that $J(n,k)$ is $(n-2)$-vertex-fault-tolerant Hamiltonian-connected.

When comparing the Johnson network $J(n,k)$ with bipartite multiprocessor topologies such as the hypercube $Q_n$ and the star graph $S_n$, a fundamental topological distinction emerges. For bipartite interconnection networks, Hamiltonian laceability serves as the natural counterpart of Hamiltonian connectivity, since the bipartition prevents Hamiltonian paths from joining arbitrary pairs of vertices. Consequently, a spanning linear array can generally be embedded only between endpoints in opposite partite sets. This parity restriction limits endpoint flexibility when arbitrary processors are required to serve as the two ends of a spanning linear structure. In contrast, $J(n,k)$ guarantees a Hamiltonian path between any two fault-free vertices, even under the fault conditions considered in this paper. This unrestricted endpoint flexibility, together with its strong fault-tolerance highlights the potential of Johnson graphs as robust interconnection topologies.

While this work establishes the tight bound for general edge faults and matching faults, the optimal upper bound for the vertex-fault tolerance of $J(n,k)$ is yet to be determined. As demonstrated by our exhaustive computer search on $J(6,3)$, the degree-minus-three edge-fault bound does not extend to vertex faults. Therefore, determining the exact, tight upper bound for the vertex-fault-tolerant Hamiltonian connectivity of general Johnson graphs remains a challenging problem. Additionally, investigating conditional edge-fault models under the minimum-degree constraint $\delta(J(n,k)-F)\ge 3$ is another natural direction for further characterizing the fault-tolerant Hamiltonian connectivity of Johnson graphs.

\vskip 0.3 in

\noindent{\bf\normalsize Data Availability} Data sharing not applicable to this article as no datasets were generated or analysed during the current study.

\vskip 0.3 in

\noindent{\bf\large Declarations}

\vskip 0.3 in

\noindent{\bf\normalsize Competing Interests} The authors have not disclosed any competing interests.

\vskip 0.3 in
\bibliographystyle{IEEEtran}
\bibliography{ref}

\end{document}